\documentclass[sigconf, nonacm]{acmart}

\usepackage{float}  
\usepackage{pgfplots}     
\pgfplotsset{compat=1.18} 
\usepackage{amsmath,amsfonts}
\usepackage{tikz}
\usetikzlibrary{shapes,arrows,positioning,calc}
\usepackage{graphicx}
\usepackage{textcomp}
\usepackage{booktabs}
\usepackage[a-1b]{pdfx}
\usepackage{subcaption}
\usepackage[dvipsnames]{xcolor}
\usepackage{colortbl}
\usepackage{bbding}
\usepackage{braket}
\usepackage{epsfig,latexsym,graphics}
\usepackage[font={small}]{caption}
\usepackage{setspace}
\usepackage{enumitem}
\pgfplotsset{compat=1.16}
\usepackage{cmd}
\usepackage[spacesep,definevectors]{easyvector} 
\usepackage{tabularx}
\usepackage{makecell}
\usepackage{multirow}
\usepackage{adjustbox}

\begin{document}

\title{BubbleRAG: Evidence-Driven Retrieval-Augmented Generation for Black-Box Knowledge Graphs}


\author{Duyi Pan}
\affiliation{%
  \institution{HKUST (GZ)}\country{China}
}
\email{dpan457@connect.hkust-gz.edu.cn}

\author{Tianao Lou}
\affiliation{%
  \institution{HKUST (GZ)}\country{China}
}
\email{tianaolou@gmail.com}

\author{Xin Li}
\affiliation{%
  \institution{HKUST (GZ)}\country{China}
}
\email{xli420@connect.hkust-gz.edu.cn}

\author{Haoze Song}
\affiliation{%
  \institution{HKUST (GZ)}\country{China}
}
\email{hsong492@connect.hkust-gz.edu.cn}

\author{Yiwen Wu}
\affiliation{%
  \institution{HKUST (GZ)}\country{China}
}
\email{ywu240@connect.hkust-gz.edu.cn}

\author{Mengyi Deng}
\affiliation{%
  \institution{HKUST (GZ)}\country{China}
}
\email{mdeng974@connect.hkust-gz.edu.cn}

\author{Mingyu Yang}
\affiliation{%
  \institution{HKUST (GZ) \& HKUST}\country{China}
}
\email{myang250@connect.hkust-gz.edu.cn}

\author{Wei Wang}
\affiliation{%
  \institution{HKUST (GZ) \& HKUST}\country{China}
}
\email{weiwcs@ust.hk}



\begin{abstract}
Large Language Models (LLMs) exhibit hallucinations in knowledge-intensive tasks. Graph-based retrieval augmented generation (RAG) has emerged as a promising solution, yet existing approaches suffer from fundamental recall and precision limitations when operating over black-box knowledge graphs---graphs whose schema and structure are unknown in advance. We identify three core challenges that cause recall loss (semantic instantiation uncertainty and structural path uncertainty) and precision loss (evidential comparison uncertainty). To address these challenges, we formalize the retrieval task as the Optimal Informative Subgraph Retrieval (OISR) problem---a variant of Group Steiner Tree---and prove it to be NP-hard and APX-hard. We propose BubbleRAG, a training-free pipeline that systematically optimizes for both recall and precision through semantic anchor grouping, heuristic bubble expansion to discover candidate evidence graphs (CEGs), composite ranking, and reasoning-aware expansion. Experiments on multi-hop QA benchmarks demonstrate that BubbleRAG achieves state-of-the-art results, outperforming strong baselines in both F1 and accuracy while remaining plug-and-play.
\end{abstract}





\maketitle

\section{Introduction}

Large Language Models (LLMs) have demonstrated remarkable capabilities in natural language processing and reasoning across diverse domains\cite{2025-surveylargelanguagemodels,2025-surveylargelanguagemodelsan,touvron2023llamaopenefficientfoundation}. Despite these advancements, LLM inference remains prone to hallucinations, and the static nature of their training data often leads to outdated knowledge\cite{2024-hallucination,2025-hallucinationinevitableinnatelimitation,2025-surveyhallucination,2025-medicalhallucinations,2024-comprehensivesurveyhallucinationmitigation}. To mitigate these limitations, Retrieval-Augmented Generation (RAG) has emerged as a critical paradigm, augmenting generation by retrieving external information at inference time\cite{2024-ragsurvey,2024-surveyretrievalaugmentedgenerationl,2021-ragnlp}.

For complex queries that span multiple documents or require multi-step reasoning chains, \emph{graph-structured} retrieval is especially beneficial. Knowledge Graphs (KGs) provide a fixcomplementary form of external knowledge that alleviates key limitations of text-centric approaches\cite{2022-surveykg,han-2026-ragvsgraphrag}. By organizing information into entities, relations, and constraints, KGs make cross-document dependencies explicit and support structured evidence composition and symbolic reasoning\cite{2025-surveykgrag,2025-surveygrag,ma-2025-llmmeetgraph,zhu-2026-surveyforgrag}. The retrieval component of such a KG-RAG system must achieve two goals simultaneously: \emph{high recall}, which retrieves all evidence relevant to the query, and \emph{high precision}, which avoids the inclusion of noisy or irrelevant context that degrades generation quality.


\begin{figure}[!t]
    \centering
    \includegraphics[width=\linewidth]{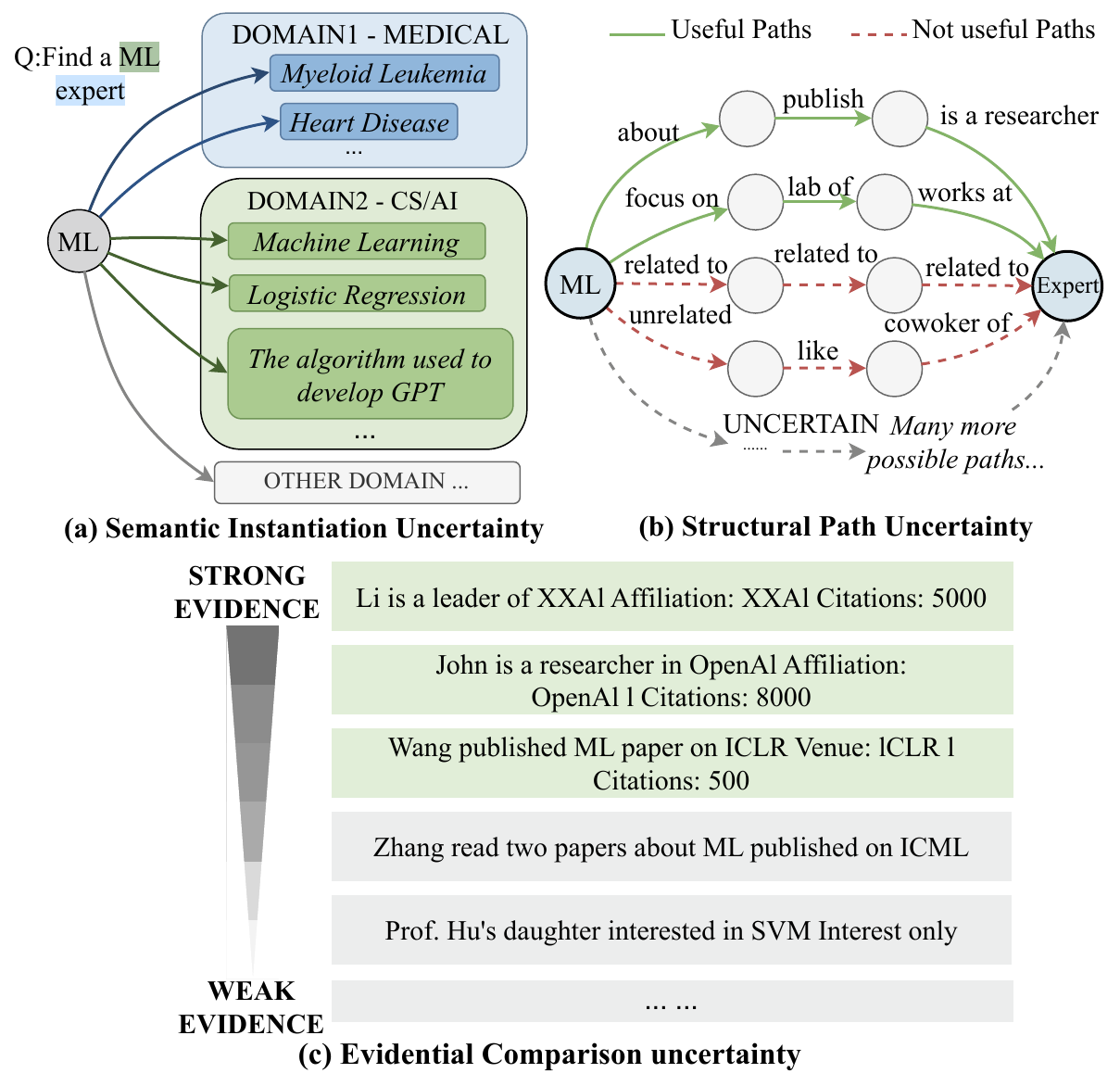}
    \caption{Three challenges in black-box Knowledge Graph retrieval that limit recall and precision: (a) Semantic instantiation uncertainty when grounding query concepts into KG entities, risking recall loss; (b) Structural path uncertainty when determining relevant relational chains, risking recall loss; and (c) Evidential comparison uncertainty when ranking candidates based on implicit evidence, risking precision loss}
    \label{fig:example}
\end{figure}

However, achieving both high recall and high precision in graph-based RAG is fundamentally difficult when the retriever operates over a \emph{black-box} KG, one whose schema, entity types, and relational structure are unknown in advance. We identify three core challenges:
(1) \textbf{Semantic instantiation uncertainty.} A query concept may be materialized in multiple heterogeneous forms, including explicit labels, aliases, attribute values, or even implicit relational patterns, so a query token does not uniquely determine where and how the concept appears in the graph. Consequently, retrieval can fail before any reasoning begins due to misalignment between query semantics and KG realizations. This directly causes \emph{recall loss}: relevant entities are missed because the retriever cannot locate them under their actual surface forms.
(2) \textbf{Structural path uncertainty.} Even when relevant entities are found, the retriever still lacks the schema knowledge needed to identify informative relational connections for the query. The same high-level relation may correspond to direct edges, multi-hop chains, or composite structures, making fixed-hop traversal or predefined patterns brittle across graphs. This again causes \emph{recall loss}: relevant relational chains are missed because the retriever applies the wrong traversal strategy.
(3) \textbf{Evidential comparison uncertainty.} When multiple candidates satisfy the constraints, the KG rarely encodes high-level notions such as expertise or importance explicitly. The retriever must instead aggregate implicit signals (e.g., publications, affiliations, citations) to support query-specific ranking, turning retrieval from graph matching into multi-factor evidence aggregation and comparative inference. This causes \emph{precision loss}: without discriminative signals, irrelevant or lower-quality candidates are ranked alongside or above the true evidence.

Together, the first two challenges cause the retriever to \emph{miss} relevant evidence, while the third prevents it from \emph{distinguishing} the most informative evidence among candidates, jointly constraining both the recall and precision of KG retrieval.

\begin{exmp}
As illustrated in Fig.~\ref{fig:example}, consider the query \textit{find an ML expert}. For the keyword \textit{ML} in a black-box KG, it could be mapped to the Medical domain or the CS/AI domain, since we do not know the KG's detailed schema. Even if we identify the right domain, it could still be mapped to different concepts, such as \textit{Logistic Regression} or \textit{The algorithm used to develop GPT.} The latter is harder to locate and may require reasoning. However, they are all concepts that represent ML. Even if we find the correct concept, we still need to determine how \textit{ML} is connected to \textit{Expert.} There could be many ways to do this within our KG, and it is hard to know the optimal linking path in advance in a black-box setting. Consequently, after finding numerous possible concepts and the linking ways between them, some will be useful while others are not. To truly find an expert, the system must assess the level of their expertise and compare them against other candidates. In this example, the first two uncertainties risk missing the correct expert entirely (recall failure), while the third risks ranking a less-qualified candidate above the true expert (precision failure).
\end{exmp}

\stitle{Prior Solutions.}
Prior work tackles the black-box difficulty through four main paradigms, yet each suffers from fundamental recall or precision limitations.
(a)~\textit{Schema-translation methods}\cite{2025-simgrag,2023-queryrewriting,2025-BYOKG-RAG} generate ideal subgraph patterns or decompose queries into pre-conceived triple structures for alignment. However, the generated patterns can only cover linking structures that the LLM already knows about, leading to limited recall when the actual graph topology diverges from the LLM's prior assumptions, particularly for distant or implicit connections.
(b)~\textit{Random-walk methods}\cite{2024-hipporag,2025-ragmemory,2025-linearrag,2025-agrag} utilize algorithms like Personalized PageRank (PPR)\cite{1999-pagerank} from semantic anchors to estimate node relevance. Yet they inherently suffer from hub bias, where high-degree nodes absorb disproportionate probability mass regardless of query relevance, and they can be inefficient when the evidence subgraph is large or spans multiple loosely connected components.
(c)~\textit{Iterative multi-hop methods}\cite{2025-levelrag,2025-graphsearch,2026-thinkongraph3,2024-thinkongraph,2025-thinkongraph2,2024-reasoninggraphs} expand from single anchor nodes via beam search or constrained breadth-first search. They are highly sensitive to initial anchor selection, and a single misalignment can cause cascading retrieval failures across the reasoning chain. In addition, they typically operate only over vertex-centric expansion, missing evidence encoded in edge labels or relational predicates.
(d)~\textit{Pre-indexed structure methods}\cite{2024-edge2024local,2024-kag,2025-cluerag} build hierarchical indices or community structures during offline pre-processing. These static, query-agnostic structures are expensive to construct and maintain, domain-dependent, and ill-suited for handling the diverse requirements of dynamic queries\cite{2025-indepthanalysisgraphbasedrag}.

None of these paradigms jointly optimizes for recall and precision within a unified retrieval objective: methods in (a)--(c) focus on finding \emph{some} relevant subgraph without systematic comparison, while (d) provides global context at the cost of query specificity.

\stitle{Our Solution.}
In this paper, we propose \textbf{BubbleRAG}, a search-based and training-free retrieval framework for black-box knowledge graphs that explicitly optimizes for both recall and precision.
To achieve high recall, BubbleRAG models the retrieval task as finding \emph{candidate evidence graphs} (CEGs), which are connected subgraphs that cover groups of semantic anchors derived from the query. Each query concept is mapped to a \emph{group} of candidate anchor nodes or edges in the KG, tolerating the heterogeneous realizations typical of black-box graphs. We formalize the task of selecting the most informative CEG as the Optimal Informative Subgraph Retrieval (OISR) problem, a variant of Group Steiner Tree, and design a heuristic Bubble Expansion algorithm to efficiently enumerate high-quality CEGs.
To achieve high precision, BubbleRAG ranks the discovered CEGs by a composite score that balances semantic relevance and structural completeness, then applies controlled LLM-guided expansion to the top-ranked candidates to refine their evidence coverage.
Unlike existing methods that treat semantic grounding, structural discovery, and evidence ranking as separable components, BubbleRAG unifies them under a single optimization-driven pipeline, directly addressing the coupled challenges of black-box KG retrieval.

\stitle{Contribution.}
We summarize our main contributions as follows:

\sstitle{Problem formulation.} We formalize KG retrieval as the OISR problem, a variant of Group Steiner Tree that captures the dual goals of semantic coverage (recall) and information density (precision). We prove that OISR is NP-hard and APX-hard, establishing the theoretical foundations for our heuristic approach.

\sstitle{Framework.} We introduce BubbleRAG, a novel search-based and training-free graph RAG pipeline for black-box knowledge graphs. It systematically addresses recall through semantic anchor grouping and heuristic bubble expansion, and precision through candidate evidence graph ranking and reasoning-aware expansion.

\sstitle{Plug-and-play Efficiency.} BubbleRAG is computationally practical and plug-and-play. It requires no retriever fine-tuning or modifications to the underlying KG structure. Its localized subgraph construction and anisotropic search ensure that the retrieval complexity remains largely independent of the global graph size, making it easily scalable to massive knowledge graphs.

\sstitle{Experiment.} Comprehensive experiments on complex multi-hop QA benchmarks (HotpotQA, MuSiQue, 2WikiMultiHopQA) demonstrate that BubbleRAG achieves state-of-the-art results. It consistently outperforms strong structure-aware baselines in both F1 score and accuracy, with the largest gains on the most challenging benchmark (MuSiQue).

\section{Preliminaries}\label{sec:preliminaries}

\subsection{Graph-Based RAG}

A graph-based RAG system typically consists of three stages: (i)~\emph{graph indexing}, (ii)~\emph{graph retrieval}, and (iii)~\emph{answer generation}. In the indexing stage, input documents are segmented into text chunks, and an LLM extracts triples to construct a knowledge graph $G=(V,E)$; the original chunks are typically retained as an additional retrieval source. In the retrieval stage, given a user query $q$, the system searches the graph to collect a subgraph $G^* \subseteq G$ that contains the evidence needed to answer $q$. In the generation stage, the LLM produces an answer conditioned on the retrieved subgraph and, optionally, the associated text chunks for grounding.

Crucially, entities and relations extracted from different chunks, potentially from different source documents, are linked into a single graph, enabling \emph{cross-document reasoning} that text-chunk retrieval alone cannot support. This distinguishes graph-based RAG from standard vector-similarity retrieval (NaiveRAG): the retriever must respect and exploit the graph's relational structure, not merely rank isolated text passages by embedding proximity.

\subsection{Black-Box Knowledge Graphs}

We call a knowledge graph \emph{black-box} if its schema, the set of entity types, relation types, and their constraints, is not provided to the retrieval system. The retriever can only access the graph's topology and the textual content associated with nodes and edges, without knowing, for instance, which relation types connect which entity types.

This is the common case in practice. LLM-extracted KGs from heterogeneous corpora have no standardized schema; the same real-world relation may appear under different labels across documents. Real-world KGs such as Wikidata have complex, evolving schemas that are difficult to expose fully to a downstream retriever. Enterprise KGs may be proprietary or partially documented. Without schema knowledge, the retriever cannot rely on predefined traversal patterns, meta-paths, or type constraints; it must reason about the graph structure \emph{at query time}.

\subsection{Motivations}

Our approach is grounded in two key characteristics of answer-supporting subgraphs.

\begin{itemize}[leftmargin=*]
    \item \textbf{Motivation~1 (Semantic Anchoring).}\label{motiva:anchor}
    An evidence subgraph capable of answering $q$ should contain nodes that align explicitly or implicitly with the key concepts in the query. Even when the terminology varies, the relevant subgraph is expected to include semantically corresponding entities or relations that can serve as anchors for retrieval.
    \item \textbf{Motivation~2 (Topological Cohesion).\label{motiva:topo}}
    Relevant evidence is not arbitrarily scattered across the graph, but tends to form a connected and compact structure. In particular, the concepts required by $q$ are likely linked through relatively tight relational paths, suggesting that the answer-supporting evidence can be captured by discovering a cohesive connecting backbone and then expanding locally as needed.
\end{itemize}

\begin{exmp}
Consider the query \textit{``When did Lothair II's mother die?''} Motivation~1 holds because the query concepts, \textit{Lothair~II}, \textit{mother}, and \textit{death date}, have corresponding anchors in the KG, even if they appear as \textit{``Lothar~II,''} a \textit{child\_of} relation, or an intermediate \textit{family} node. Motivation~2 holds because these anchors and the answer entity are connected through a compact subgraph (e.g., Lothair~II $\xrightarrow{\text{mother}}$ Gisela $\xrightarrow{\text{died}}$ 860~AD), rather than being scattered across disconnected components.
\end{exmp}

Together, these motivations suggest a retrieval strategy that (i)~identifies candidate anchor nodes for each query concept (Motivation~1), (ii)~discovers compact connected subgraphs spanning these anchors (Motivation~2), and (iii)~ranks the resulting \emph{candidate evidence graphs} (CEGs) by their relevance to the query. We formalize the first two requirements in the following definition.

\begin{table}[!t]
  \caption{Summary of Notations}\label{tab:notations}\vspace{-2ex}
  \small
  \begin{tabular*}{\linewidth}{@{\extracolsep{\fill}} p{20mm} | p{68mm}}
    \toprule
    Notation & Description \\
    \midrule
    $G = (V, E)$ & Knowledge graph with nodes $V$, edges $E$ \\
    $\mathrm{val}(\cdot)$ & Value function for nodes and edges \\
    $q$ & Input user query \\
    $t_i$ & Keyword extracted from query $q$ \\
    $\mathcal{S} = \{S_1, \dots, S_m\}$ & Semantic anchor groups \\
    $w_i$ & Importance weight of group $S_i$ \\
    $G^* = (V^*, E^*)$ & Optimal informative subgraph (target) \\
    $G' = (V', E')$ & Localized search space subgraph \\
    $h$ & Hop threshold for subgraph extraction \\
    $B$ & Search budget for bubble expansion \\
    $\mathrm{Cost}_{\mathrm{sem}}(T)$ & Semantic dissonance cost of CEG $T$ \\
    $r_{\mathrm{miss}}$ & Missing mass (uncovered group weights) \\
    $\alpha$ & Penalty factor for structural completeness \\
    $\mathrm{Penalty}_{\mathrm{miss}}(T)$ & Structural incompleteness penalty \\
    $n$ & Number of top CEGs selected \\
    $d$ & Max depth for multi-hop expansion \\
    $G_{\mathrm{final}}$ & Unified evidence graph merging CEGs \\
    $C_{\mathrm{text}}$ & Text chunks linked with evidence graph \\
    \bottomrule
  \end{tabular*}
\end{table}

\subsection{Problem Formulation}

Motivations~1--2 motivate a retrieval formulation where the goal is to find a connected subgraph of $G$ that (a)~covers a set of anchor groups derived from the query concepts, and (b)~maximizes information density. Let $\mathcal{S} = \{S_1, S_2, \dots, S_k\}$ be a collection of \emph{semantic anchor groups}, where each $S_i$ contains candidate nodes or edges corresponding to the $i$-th query concept. Let $\mathrm{val}(\cdot)$ be a value function measuring the semantic relevance of a node or edge to the query. We assume $G$ is \emph{informationally sufficient} for $q$, i.e., at least one connected subgraph satisfying the coverage constraints below exists.

\begin{defn}[Optimal Informative Subgraph Retrieval (OISR)]\label{defn:OISR}
~\\
\begin{itemize}[leftmargin=1.5em]
    \item \textbf{Input:}
    A knowledge graph $G=(V,E)$, where each node $v \in V$ and each edge $e \in E$ is associated with a value function $\mathrm{val}(\cdot)$, and a collection of anchor sets $\mathcal{S} = \{S_1, S_2, \dots, S_k\}$, where each $S_i \subseteq (V \cup E)$.

    \item \textbf{Output:}
    A subgraph $G'=(V',E')$ such that $G' \subseteq G$.

    \item \textbf{Constraints:}
    \begin{itemize}
        \item \textbf{Connectivity:} $G'$ is a connected graph.
        \item \textbf{Multi-Set Coverage:} For each anchor set $S_i \in \mathcal{S}$, the retrieved subgraph must contain at least one element from that set, i.e., $(V' \cup E') \cap S_i \neq \emptyset, \forall i \in \{1, \dots, k\}$.
    \end{itemize}

    \item \textbf{Objective:}
    Maximize the average value of the nodes and edges in the selected subgraph:
    \begin{equation}
        \max_{G' \subseteq G} \Phi(G') = \frac{\sum_{v \in V'} \mathrm{val}(v) + \sum_{e \in E'} \mathrm{val}(e)}{|V'| + |E'|}
    \end{equation}
\end{itemize}
\end{defn}

The average-value objective encodes an \emph{information density} principle: it rewards compact subgraphs whose nodes and edges are all query-relevant, while penalizing subgraphs that achieve coverage through long, irrelevant paths. Total value would favor sprawling subgraphs; minimum value would be overly conservative and brittle to noise. Note that anchor sets may include edges ($S_i \subseteq V \cup E$), reflecting our system's ability to match query concepts against both entities and relations (see \S\ref{sec:bubble}).

We prove in \S\ref{sec:ceg-discovery} that OISR is NP-hard and APX-hard (Theorems~\ref{thm:nphard}--\ref{thm:apxhard}), motivating the heuristic approach developed in Section~\ref{sec:bubble}.



\begin{figure*}[t]
    \centering
    \includegraphics[width=\linewidth]{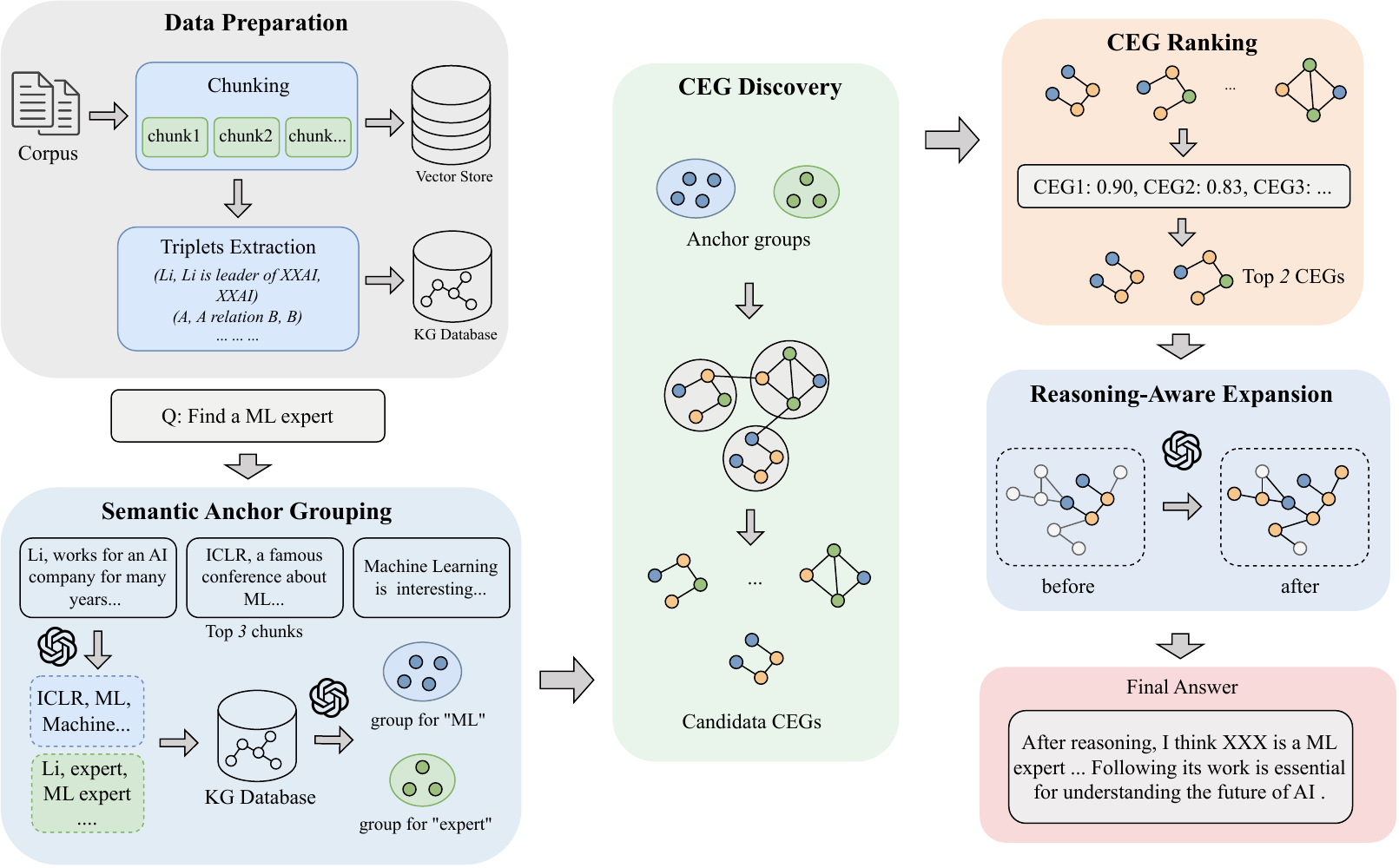}
    \caption{Pipeline of BubbleRAG.}
    \label{fig:pipeline}
\end{figure*}

\section{The BubbleRAG Framework}\label{sec:bubble}

BubbleRAG consists of an offline data preparation phase followed by four online retrieval stages, as illustrated in Figure~\ref{fig:pipeline}. The complete pipeline is summarized as follows:

\begin{itemize}[leftmargin=*, itemsep=0.3em, topsep=0.3em]
    \item \textbf{Step 1 (Data Preparation, \S\ref{sec:data-prep}):} Build a knowledge graph from the text corpus with semantically enriched edges, ensuring that both entities and relations carry matchable content for downstream retrieval.
    \item \textbf{Step 2 (Semantic Anchor Grouping, \S\ref{sec:anchor-group}):} Extract keywords from the query, find matching KG nodes/edges, and organize them into weighted anchor groups---serving recall by tolerating aliases and schema variations.
    \item \textbf{Step 3 (CEG Discovery, \S\ref{sec:ceg-discovery}):} Connect anchor groups via the Bubble Expansion heuristic to enumerate candidate evidence graphs (CEGs)---serving recall by discovering connected evidence structures.
    \item \textbf{Step 4 (CEG Ranking, \S\ref{sec:ceg-ranking}):} Score candidates by semantic relevance and structural completeness---serving precision by filtering out noisy or incomplete candidates.
    \item \textbf{Step 5 (Reasoning-Aware Expansion, \S\ref{sec:reasoning-expand}):} Apply LLM-guided multi-hop expansion to top-ranked CEGs---serving precision by refining evidence coverage for the most promising candidates.
\end{itemize}

We detail each step below.

\subsection{Data Preparation}\label{sec:data-prep}

Following prior work on LLM-based knowledge graph construction \cite{2025-lightrag,2024-hipporag}, BubbleRAG builds its knowledge graph using a standard pipeline: (1) chunking the corpus into passages, (2) using an LLM to extract triplets from each passage, and (3) indexing the triplets into a graph structure. Crucially, entities appearing in different chunks, potentially from different source documents, are linked into a single graph, enabling cross-document retrieval that text-chunk approaches cannot support.

A key distinction of BubbleRAG lies in its edge representation. Rather than treating relations as simple labels between entities, we embed rich text directly into the edges. Specifically, for a standard triplet $(A, R, B)$, the edge connecting node $A$ and node $B$ stores the combined textual content $ARB$. This enrichment enables edge-level semantic matching during anchor grouping (\S\ref{sec:anchor-group}), allowing the system to match query concepts against both entities and relations, which is essential for handling queries where the key constraint is a relational predicate (e.g., \textit{authored by,} \textit{died in}).

\subsection{Semantic Anchor Grouping}\label{sec:anchor-group}


This step operationalizes Motivation~1 (Semantic Anchoring) by ensuring that each query concept is mapped to a \emph{group} of candidate anchors, tolerating the heterogeneous realizations typical of black-box KGs. The group-based design directly serves recall: rather than committing to a single best-matching node (which may be wrong), we maintain multiple candidates per concept.

The primary objective of this stage is to extract necessary keywords, which covering both core concepts and relations---from a user query and map them to \emph{anchors} within the knowledge graph. A \emph{semantic anchor group} $S_i$ is a set of KG nodes or edges that are candidate realizations of the $i$-th query concept. The collection $\mathcal{S} = \{S_1, \ldots, S_m\}$, together with importance weights $\{w_1, \ldots, w_m\}$, forms the input to the next phase.

Given a query $q$, the system first retrieves the top-$k$ relevant text chunks from the corpus. If these chunks contain sufficient information to directly answer $q$, BubbleRAG terminates the graph traversal and outputs the response. If the information is insufficient, these chunks act as contextual evidence for the subsequent graph anchor extraction.


\stitle{Keyword Extraction and Latent Inference.}
A naive approach to obtaining keywords relies on Named Entity Recognition (NER), treating extracted surface mentions as direct query signals. However, real-world user queries frequently contain latent information critical for retrieval that is not explicitly stated, such as implied concepts, implicit relations, or underspecified constraints. BubbleRAG overcomes this by leveraging the reasoning capabilities of LLMs to go beyond literal matching: given a query $q$, the LLM extracts not only explicit keywords but also infers implicit yet necessary concepts. This facilitates more accurate mapping to anchors in a black-box knowledge graph, thereby reducing the reliance on long, sparse multi-hop traversals and enabling the system to target the relevant answer subgraph more directly. By inferring latent keywords, the system recovers query concepts that would be invisible to surface-level extraction, preventing recall loss from implicit references.

\begin{exmp}
For the query \textit{``Find scientific papers authored by the winner of the 1921 Nobel Prize in Physics,''} a simple extractor might only identify \textit{Nobel Prize}, \textit{1921}, and \textit{scientific papers.} BubbleRAG, instead, infers that \textit{Albert Einstein} is the key anchor concept for the query.
\end{exmp}


\stitle{Anchor Specialization.}
After extracting keywords $\{t_i\}$ from $q$, BubbleRAG retrieves a top-$k$ pool of candidate nodes for each $t_i$. However, many keywords become \emph{underspecified} when used as standalone query signals. A generic keyword like \textit{mother} matches thousands of nodes globally, but only a few are relevant to the specific query context. We prompt the LLM to rewrite each underspecified keyword into a query-conditioned constraint, explicitly binding it to other anchor intents already present in $q$. For instance, instead of retrieving candidates for the generic term \textit{mother}, BubbleRAG refines the signal into \textit{Lothair II's mother}. This forces the similarity search to prioritize locally compatible candidates, improving precision (filtering global false positives) without sacrificing recall (the refined query is strictly more specific, so any true match still qualifies).

\begin{exmp}
For the query \textit{``When did Lothair II's mother die?''}, the keyword \textit{Lothair II} typically maps to a tight candidate pool near the intended subject region. Conversely, the standalone keyword \textit{mother} retrieves a flood of generic ``mother-of'' relations and maternal entities scattered across the graph. Crucially, many of these top-$k$ matches fall outside the vicinity of the \textit{Lothair II} anchor candidates, thereby consuming the limited retrieval budget without improving the likelihood of forming a high-quality anchor group.
\end{exmp}


\stitle{Schema Relaxation.}
While anchor specialization tightens the search, the opposite problem can also arise: a query concept may not match any KG label exactly due to schema differences. Before finalizing the graph query, BubbleRAG uses the text chunks retrieved in the pre-retrieval step as previews of local communities within the knowledge graph. A retrieved chunk often reveals the co-occurrence of multiple keyword intents, providing preliminary evidence that certain anchors are locally connected in a promising region. Guided by this preview, BubbleRAG can strategically relax schema-sensitive conditions to preserve recall: if a chunk confirms the presence of highly relevant neighbors, a strict structural condition may be relaxed to a necessary (but not sufficient) one. Schema relaxation prevents recall loss caused by rigid label matching in KGs with non-standardized schemas.

\begin{exmp}
Consider the query \textit{``Find the second-married woman whose son started a war.''} A strict search for a relation labeled \textit{``second marriage''} may fail if the graph merely uses a generic \textit{``spouse''} label. If a retrieved chunk states \textit{``...she married X, and later her son Y launched an invasion...''}, BubbleRAG detects the high relevance of this region and allows \textit{``second marriage''} to be relaxed to \textit{``marriage''}. Although \textit{``marriage''} is globally noisy, within this specific community, it is highly likely to be the correct anchor because the required contextual neighbors (son, war) are confirmed to be nearby.
\end{exmp}


\stitle{Anchor Grouping and Importance Weighting.}
Following the retrieval of top-$k$ candidates, we accumulate an initial pool of potential anchors. Since distinct keywords may map to overlapping concepts, BubbleRAG invokes an LLM to perform \emph{semantic anchor grouping}, merging candidates that refer to the same underlying intent into cohesive groups, denoted as $\mathcal{S}=\{S_1,S_2,\dots,S_m\}$. Concretely, the LLM receives the full set of retrieved candidate nodes along with the query, and clusters them into groups based on which query concept they correspond to. Candidates matching the same concept (possibly via different keywords) are merged into one group.

In real-world knowledge graphs, perfect coverage of all query elements is rarely guaranteed; some concepts may be entirely missing or structurally disconnected. BubbleRAG therefore permits partial alignment while establishing a mechanism to evaluate match quality. The LLM assigns a normalized importance weight $w_i$ to each group $S_i$ ($\sum_{i=1}^{m} w_i = 1$). The LLM assigns weights based on centrality to the query's answering logic: the core subject entity typically receives the highest weight, while peripheral modifiers (e.g., temporal constraints, relational predicates) receive lower weights. For instance, in \textit{When did Lothair II's mother die?}, \textit{Lothair~II} is indispensable (high weight $\approx 0.5$), \textit{mother} is a necessary relation (medium weight $\approx 0.3$), while \textit{death date} is the target attribute (lower weight $\approx 0.2$). Missing the core subject is catastrophic; missing the temporal modifier may still allow partial retrieval. The weights enable graceful degradation: the ranking step (\S\ref{sec:ceg-ranking}) penalizes missing high-weight groups severely but tolerates missing low-weight groups, allowing the system to retrieve \emph{partially informative} subgraphs rather than returning nothing.

\subsection{Candidate Evidence Graph Discovery}\label{sec:ceg-discovery}


Following semantic anchor grouping, the query is represented by a collection of weighted anchor groups $\mathcal{S}=\{S_1,S_2,\dots,S_m\}$. To operationalize Motivation~2 (Topological Cohesion), we now seek \emph{candidate evidence graphs} (CEGs)---connected subgraphs of $G$ that cover these anchor groups while maximizing information density. This corresponds to the OISR problem defined in \S\ref{sec:preliminaries}.


\stitle{Problem Hardness.}
To determine the appropriate algorithmic strategy, we analyze the computational complexity of OISR. We prove that OISR can be reduced from the NP-hard Group Steiner Tree (GST) problem\cite{2003-plylog,2009-Karp-NP-NPC-NP-Reduce}.

\begin{thm}\label{thm:nphard}
The Optimal Informative Subgraph Retrieval problem is NP-hard.
\end{thm}

\textit{Proof.} We define the decision version of OISR as follows: Given a graph $G=(V,E)$, a value function $\mathrm{val}(\cdot)$, a collection of anchor sets $\mathcal{S} = \{S_1, \dots, S_k\}$, and a threshold $\alpha \in \mathbb{R}^+$, does there exist a connected subgraph $G'=(V',E') \subseteq G$ such that:
\begin{enumerate}[label=(\roman*)]
    \item $(V' \cup E') \cap S_i \neq \emptyset$ for all $i \in \{1, \dots, k\}$, meaning all anchor sets are covered;
    \item The average value satisfies $\frac{\sum_{v \in V'} \mathrm{val}(v) + \sum_{e \in E'} \mathrm{val}(e)}{|V'| + |E'|} \ge \alpha$.
\end{enumerate}

Now, consider an arbitrary unweighted Group Steiner Tree (GST) instance with graph $G(V, E)$, node groups $\mathcal{S}_{GST} = \{S_1, \dots, S_k\}$ where $S_i \subseteq V$, and a total size limit $\gamma$. We construct an OISR instance as follows:
\begin{itemize}
    \item Use the same graph $G$ as the knowledge graph.
    \item Construct the anchor sets as $\mathcal{S} = \{S_1, \dots, S_k\} \cup \{S_0\}$, where $S_0 = \{v^*\}$ for an arbitrary node $v^* \in V$. (The choice of $v^*$ does not affect the reduction's correctness.)
    \item Set $\mathrm{val}(v^*) = 1$, and set $\mathrm{val}(x) = 0$ for all other nodes $x \in V \setminus \{v^*\}$ and all edges $e \in E$.
    \item Set the threshold $\alpha = \frac{1}{\gamma}$.
\end{itemize}

Because any feasible subgraph $G'$ must cover $S_0$, it must contain $v^*$. Since $v^*$ is the only element with a non-zero value, the numerator of the objective function is exactly $1$. Condition (ii) thus becomes:
\begin{equation*}
    \frac{1}{|V'| + |E'|} \ge \frac{1}{\gamma} \implies |V'| + |E'| \le \gamma
\end{equation*}

Therefore, a Group Steiner tree of size $\le \gamma$ exists in the GST instance if and only if a feasible subgraph exists in the OISR instance. This establishes a polynomial-time reduction, proving that OISR is NP-hard. \hfill$\square$

\begin{thm}\label{thm:apxhard}
The OISR problem is hard to approximate.
\end{thm}

\textit{Proof.} We show that a constant-factor approximation for OISR would imply a constant-factor approximation for GST, contradicting established inapproximability bounds.
Using the approximation-preserving nature of the reduction established above, let $OPT_{GST}$ be the size of the optimal Group Steiner Tree for the given GST instance. Based on our construction, the optimal objective value for the OISR instance is exactly $OPT_{OISR} = \frac{1}{OPT_{GST}}$.

Suppose there exists a polynomial-time algorithm that achieves a constant approximation ratio $\rho > 0$ for OISR. Let $ALG_{OISR}$ be the objective value obtained by this algorithm, which satisfies:
\begin{equation}
    ALG_{OISR} \ge \rho \cdot OPT_{OISR}
\end{equation}

By the construction of our OISR instance, any retrieved subgraph yielding an objective value of $ALG_{OISR}$ corresponds to a feasible Group Steiner Tree of size $ALG_{GST} = \frac{1}{ALG_{OISR}}$. Substituting this relation into the inequality gives:
\begin{equation}
    \frac{1}{ALG_{GST}} \ge \rho \cdot \frac{1}{OPT_{GST}} \implies ALG_{GST} \le \frac{1}{\rho} \cdot OPT_{GST}
\end{equation}

Since $\rho$ is a constant, this implies the existence of a polynomial-time algorithm that approximates the Group Steiner Tree problem within a constant factor of $\frac{1}{\rho}$. However, it is a well-established result that GST on general graphs cannot be approximated within a factor of $O(\log^{2-\epsilon} k)$ unless $\mathrm{NP} \subseteq \mathrm{ZTIME}(n^{\mathrm{polylog}(n)})$\cite{2003-plylog}. This contradiction proves that no such constant-factor approximation algorithm can exist for OISR, rendering the problem APX-hard. Since this hardness holds even for the restricted class of binary-valued instances (val $\in \{0,1\}$), it holds a fortiori for general OISR. \hfill$\square$

Since no efficient exact or constant-factor approximation algorithm exists, we design a targeted heuristic to generate high-quality CEGs in practice.


\stitle{Heuristic Solution: Bubble Expansion.}
We term this heuristic \emph{Bubble Expansion} because the cost-guided wavefronts expanding from each anchor group resemble inflating bubbles: they grow anisotropically through low-resistance (query-aligned) regions and merge upon collision with bubbles from other groups.

To bridge the OISR objective with a practical graph traversal algorithm, we instantiate the value function as $\mathrm{val}(v) = \cos(z_q, z_v)$ and define a node cost:
\begin{equation}\label{eq:cost}
    \text{cost}(v) = 1 - \mathrm{val}(v) = 1 - \cos(z_q, z_v)
\end{equation}

\begin{figure*}[t]
    \centering
    \includegraphics[width=\linewidth]{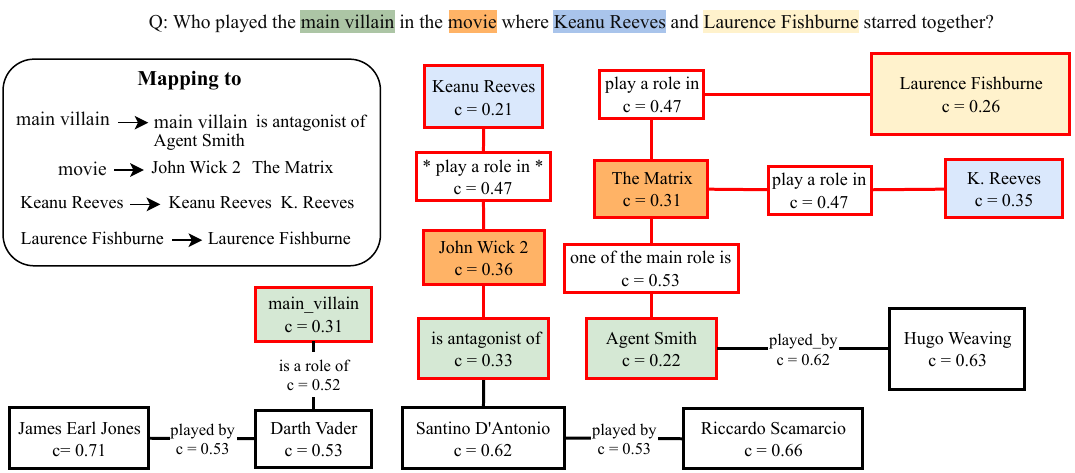}
    \caption{An example of the Candidate Evidence Graph (CEG) generated by Bubble Expansion. Different highlight colors in the query represent different extracted keywords, which are mapped to their corresponding semantic anchor groups in the graph. Elements with red borders indicate the nodes and edges that have been successfully incorporated into a CEG. }
    \label{fig:pipeline}
\end{figure*}

By driving the traversal via the accumulation of this semantic cost, the algorithm operationalizes the information density objective to rapidly discover query-aligned structural skeletons. The search consists of three phases:

\begin{enumerate}[leftmargin=*, itemsep=0.3em, topsep=0.3em]
    \item \textbf{Localized Graph Construction.}
    To constrain the search space, we first extract a localized subgraph $G'=(V',E')$. Starting from the anchor nodes in each semantic group $S_i$, we collect their $h$-hop neighborhoods (where parameter $h$ controls the localization range) and take the union of these subgraphs to form $G'$. This step substantially reduces computation by focusing exclusively on regions topologically reachable from the anchors within a limited budget.

    \item \textbf{Anisotropic Expansion (Multi-source, Cost-guided Search).}
    Within $G'$, we initiate a multi-source, Dijkstra-like expansion from all anchor nodes. Unlike standard Breadth-First Search (BFS) that expands uniformly by hop count, our expansion is \emph{anisotropic}: it propagates preferentially through nodes with a low \emph{accumulated semantic cost}. Concretely, for each node, we maintain: (i) the minimum accumulated cost to reach it from a given source group, (ii) a predecessor pointer for path reconstruction, and (iii) a bitmask recording which semantic anchor groups have reached it. The search naturally advances rapidly through low-resistance (query-aligned) regions while stalling in high-resistance (irrelevant) areas.

    \item \textbf{Collision Detection and Subgraph Fusion.}
    A critical event occurs when an expansion frontier reaches a node already covered by \emph{different} semantic anchor groups (i.e., its bitmask indicates multi-group coverage). We treat such a node as a candidate \textbf{Steiner node} (meeting point) and trigger a backtracing procedure: we reconstruct the low-cost paths from this node back to the involved anchor nodes using the predecessor pointers. These paths are subsequently fused into a connected CEG and added to the candidate set for final ranking.
\end{enumerate}


\stitle{Handling Intra-Group and Failed Connections.}
To ensure robustness in real-world scenarios, our heuristic accounts for two edge cases:
\begin{itemize}[leftmargin=*, itemsep=0.3em, topsep=0.3em]
    \item \textbf{Intra-Group Connections.} BubbleRAG permits connections among anchor nodes within the \textit{same} semantic group. Strong intra-group connectivity often indicates that relevant evidence is highly concentrated locally. We allow such fusions into the candidate set but apply a mild penalty during the ranking phase to encourage inter-group coverage and prevent a single group from dominating the retrieved subgraph.
    \item \textbf{Fallback Mechanism.} In exceptionally sparse graphs, anchor nodes may fail to collide within the expansion budget. Instead of returning an empty set, BubbleRAG falls back to the currently identified anchor nodes as a baseline. Although disconnected, this preserves non-empty evidence for subsequent reasoning.
\end{itemize}


\begin{algorithm}[t]
\DontPrintSemicolon
\SetKwInOut{Input}{Input}
\SetKwInOut{Output}{Output}

\Input{Knowledge graph $G=(V,E)$, semantic anchor groups $\mathcal{S}=\{S_1,\dots,S_m\}$, localization hop $h$, expansion budget $B$}
\Output{Set of candidate subgraphs $\mathcal{C}$}

\BlankLine
\textbf{// Phase 1: Localized Graph Construction}
$G' \gets \emptyset$\;
\ForEach{$S_i \in \mathcal{S}$}{
    $N_i \gets$ $h$-hop neighborhood of nodes in $S_i$\;
    $G' \gets G' \cup N_i$\;
}
$G' \gets (V', E')$ where $V',E'$ are nodes and edges in the union\;

\textbf{// Phase 2: Anisotropic Expansion}
Initialize priority queue $PQ$ with all terminals $t \in \bigcup_i S_i$\;
Initialize $cost[t] \gets 0$, $mask[t] \gets$ group bitmask of $t$, $pred[t] \gets \emptyset$\;
\While{$PQ$ is not empty and $|C| < B$}{
    $(v, c) \gets PQ.pop()$ with minimum accumulated cost\;
    \If{$c > cost[v]$}{continue\;}

    \ForEach{neighbor $u$ of $v$ in $G'$}{
        $new\_cost \gets c + cost(u)$\;
        $new\_mask \gets mask[v]$ OR group bitmask of $u$\;

        \If{$new\_cost < cost[u]$}{
            $cost[u] \gets new\_cost$\;
            $mask[u] \gets new\_mask$\;
            $pred[u] \gets v$\;
            $PQ.push((u, new\_cost))$\;

            \textbf{// Collision Detection}
            \If{$new\_mask$ covers multiple groups}{
                $T \gets$ backtrack paths from $u$ to terminals using $pred$\;
                $C \gets C \cup \{T\}$\;
            }
        }
    }
}

\textbf{// Phase 3: Fallback Handling}
\If{$C$ is empty}{
    $C \gets$ all terminals\;
}
\Return{$\mathcal{C}$}

\caption{Bubble Expansion Algorithm}
\label{algo:bubble}
\end{algorithm}


\stitle{Complexity Analysis.}
We analyze the computational complexity of the CEG discovery phase across its two main steps:

\begin{itemize}[leftmargin=*, itemsep=0.3em, topsep=0.3em]
    \item \textbf{Localized Graph Construction.}
    Given $m$ semantic anchor groups comprising $n = \sum_{i=1}^{m} |S_i|$ anchor nodes, we extract the $h$-hop neighborhood for each anchor. The maximum number of nodes reachable within $h$ hops from a single anchor is bounded by $O(d_{\text{avg}}^h)$, where $d_{\text{avg}}$ is the average node degree. Thus, the localized subgraph $G' = (V', E')$ satisfies $|V'| \leq n \cdot O(d_{\text{avg}}^h)$. For sparse LLM-constructed knowledge graphs and a small $h$, this step runs efficiently in $O(n \cdot d_{\text{avg}}^h)$ time.

    \item \textbf{Anisotropic Bubble Expansion.}
    The expansion maintains a state for each node $v \in V'$ containing: (i) minimum accumulated costs from each source group, (ii) predecessor pointers, and (iii) a bitmask of size $m$. The worst-case time complexity, implemented via a priority queue, is $O(|E'| \cdot m \cdot \log|V'|)$. In practice, the anisotropic nature of the search prunes high-cost paths early, keeping the empirical running time significantly below the worst-case bound. Collision detection and subgraph fusion run in $O(|V'| \cdot m)$ time via bitmask overlap checks.

    \item \textbf{Overall Complexity.}
    Combining both steps, the total time complexity is:
    \[
    O\big(n \cdot d_{\text{avg}}^h + |E'| \cdot m \cdot \log|V'| + |V'| \cdot m\big)
    \]
    Crucially, because $|V'| \ll |V|$ due to the localized construction, the retrieval complexity per query is largely independent of the global graph size $|V|$. This ensures that BubbleRAG scales efficiently to massive knowledge graphs. In our experiments, the localized subgraph $G'$ typically contains fewer than $10^3$ nodes even when the full KG has $10^5$+ nodes, confirming this scalability.
\end{itemize}

Moreover, the memory footprint is dominated by the localized subgraph $G'$ and the expansion state, requiring $O(|V'| + |E'| + m \cdot |V'|)$ space. This remains highly tractable for graphs with millions of nodes, provided the localization parameter $h$ is properly configured.

\subsection{Candidate Evidence Graph Ranking}\label{sec:ceg-ranking}


A key design choice in BubbleRAG is the \emph{decoupling} of CEG discovery (\S\ref{sec:ceg-discovery}) from CEG ranking. The discovery phase generates candidates using a fast graph-based heuristic with a simple cost function. The ranking phase then re-evaluates these candidates using a richer criterion that considers both semantic relevance and structural completeness. This decoupling offers two benefits: (a) the heuristic can explore broadly without being slowed by an expensive scoring function, and (b) the scoring function is \emph{plug-and-play}, meaning it can be replaced or augmented without changing the discovery algorithm.

Bubble Expansion in the previous phase generates a set of CEGs.
While these candidates are structurally connected, connectivity alone does not
guarantee evidentiary quality: a candidate may bridge anchors via a generic hub,
yielding a connected but semantically vacuous structure. Conversely, a
high-quality candidate may miss a minor keyword yet still capture the core
intent.

We therefore introduce a discriminative ranking mechanism that integrates
semantic dissonance and structural incompleteness. To align with
intuitive ranking (higher is better), we define the final score as the inverse
of a composite cost:
\begin{equation}
\text{Score}(T) = \frac{1}{\text{Cost}_{\text{sem}}(T)\cdot \text{Penalty}_{\text{miss}}(T) + \epsilon},
\label{eq:score}
\end{equation}
where $\epsilon$ is a small constant to prevent division by zero. The semantic dissonance approximates $1 - \Phi(G')$ (the inverse of the average value in the OISR objective), while the missing-group penalty enforces the coverage constraint with graceful degradation.

\stitle{Semantic dissonance.}
This component measures how well the nodes in a CEG align with the query. A high-quality CEG should contain nodes that are all semantically close to the query, not just a few relevant nodes diluted by many irrelevant ones.
For a node $v$, let $\text{cost}(v)=1-\cos(z_q,z_v)$. The semantic
cost of $T$ is:
\[
\text{Cost}_{\text{sem}}(T) = \frac{1}{|V_T|} \sum_{v \in V_T} \text{cost}(v).
\]
We use the mean rather than the sum to ensure \textit{size invariance}: the discovery phase
is search-oriented and favors low-resistance connectivity, whereas ranking
compares candidates of potentially different sizes. Averaging prevents
information-rich, longer evidence chains from being penalized solely for containing more nodes, as long as those nodes
remain query-relevant (low cost).

\stitle{Structural incompleteness.}
This component enforces the coverage requirement from the OISR formulation. A CEG that covers all anchor groups is strongly preferred over one that misses a group, especially a high-weight group.
Let $w_i$ be the importance weight of group $S_i$ (e.g.,
$w_{\text{Lothair}}=0.8$, $w_{\text{Mother}}=0.2$). The missing mass is
$r_{\text{miss}}=\sum_{i: S_i \cap V_T = \varnothing} w_i$, and the penalty is:
\begin{equation}
\text{Penalty}_{\text{miss}}(T) = e^{\alpha \cdot r_{\text{miss}}}.
\label{eq:penalty}
\end{equation}
This implements a weighted tolerance strategy: with a sufficiently large
$\alpha$, missing a critical group (high $w_i$) sharply increases the penalty
and strongly suppresses the score, whereas missing a peripheral group incurs
only a mild penalty. The hyperparameter $\alpha$ controls the strictness of the completeness
requirement.

\stitle{Supporting Diverse Query Semantics.}
The penalty factor $\alpha$ provides a principled mechanism to support different query semantics beyond the default balanced setting:
\begin{itemize}[leftmargin=*, itemsep=0.3em, topsep=0.3em]
    \item \textbf{AND queries} (all concepts required): A large $\alpha$ (e.g., $\alpha \ggg 1$) makes $e^{\alpha \cdot r_{\text{miss}}}$ grow sharply whenever any anchor group is uncovered, forcing the ranking to strongly prefer CEGs that cover \emph{all} query concepts simultaneously, effectively implementing conjunction semantics.
    \item \textbf{OR queries} (any concept sufficient): Setting $\alpha \approx 0$ collapses the penalty to $1$ regardless of missing groups, so ranking is driven purely by $\text{Cost}_{\text{sem}}(T)$. CEGs matching \emph{any} subset of anchor groups are treated equally in terms of coverage, effectively implementing disjunction semantics.
    \item \textbf{Comparison queries} (rank multiple candidates against each other): BubbleRAG naturally supports comparison queries through its top-$n$ CEG selection. When a query asks to compare entities (e.g., \textit{Who is more influential: Author A or Author B?}), each top-ranked CEG captures evidence centered on a different candidate. The $n$ selected CEGs are passed together to the reasoning-aware expansion and answer generation stages, where the LLM can directly compare evidence subgraphs side by side, with no architectural change required.
\end{itemize}
This flexibility allows BubbleRAG to handle a broad spectrum of query types by adjusting a single hyperparameter, without modifying the retrieval pipeline.

Based on $\text{Score}(T)$, we select the top-$n$ CEGs (parameter $n$ controls the number of selected candidates) as evidentiary foundations for the subsequent reasoning phase.

\begin{figure}[t]
    \centering
    \includegraphics[width=\linewidth]{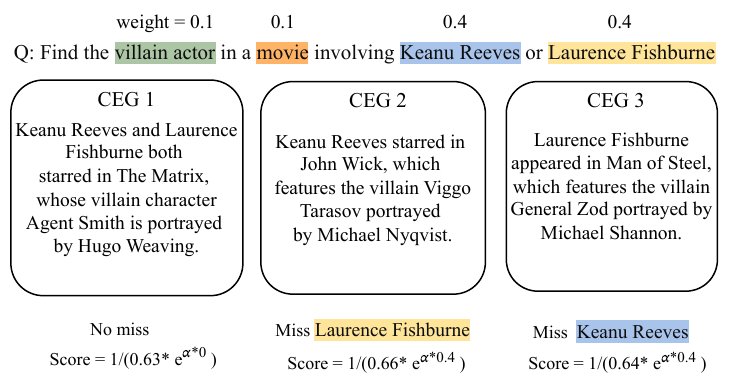}
    \caption{An example of Candidate Evidence Graph (CEG) Ranking. Different highlight colors in the query represent different semantic anchor groups, each assigned an importance weight. The final score of a CEG is determined by two components: its average semantic cost and a structural incompleteness penalty derived from missing concept groups. By simply adjusting the hyperparameter $\alpha$, the system can dynamically alter the ranking results to support AND operation and OR operation. The algorithm naturally supports compare queries by selecting the top-$n$ CEGs.}
    \label{fig:pipeline}
\end{figure}

\subsection{Reasoning-Aware Expansion}\label{sec:reasoning-expand}


The CEGs discovered in \S\ref{sec:ceg-discovery} represent minimal connected backbones spanning the anchor groups. However, the final answer entity may not be \emph{within} this backbone but in its \emph{immediate neighborhood}, one or two hops beyond the backbone's boundary. More generally, the backbone captures the \emph{reasoning chain} leading to the answer, while the answer itself requires one final inference step. This pattern arises whenever the query asks about a property of an entity that must itself be identified through intermediate reasoning.

For example, consider: \textit{Who are the lead actors in the movies
co-directed by Director A and Director B?} A ranked subgraph may already
bridge the two director anchors and identify their collaborative work (e.g.,
Movie X). Yet the user intent targets the actors of this movie,
which are neighbors of Movie X but were not part of the initial anchor
set. This pattern is general: the backbone provides the reasoning chain, while expansion uncovers the target property.

This step serves precision by focusing LLM reasoning on the most promising CEGs. Rather than expanding all candidates (expensive), we apply expansion only to the top-$n$ ranked CEGs, using the LLM's judgment to select the most relevant neighbors at each hop.

To address this, BubbleRAG performs a controlled, LLM-guided multi-hop
expansion. Starting from each selected candidate $T^*$, we iteratively expand
the current frontier up to a maximum depth $d$ (parameter $d$ controls the reasoning depth). At each hop $\ell\in[1,d]$, we
retrieve the immediate neighbors $\mathcal{N}_d(T^*)$ adjacent to the current
frontier, and prompt the LLM with the query and current evidence to select the most
promising neighbors. In the example above, once Movie X is identified,
the LLM preferentially selects neighbors connected by cast-related relations
(e.g., Actor Y), while discarding irrelevant attributes such as release
date or production company. The selected nodes/edges are then added to the
evidence structure:
\[
T^*_{d} = T^*_{d-1} \cup \text{Selected}\!\left(\mathcal{N}_d\right).
\]

The expansion continues until the depth limit $d$ is reached, or terminates
early when the LLM selects no new neighbors (or indicates that the current
evidence is sufficient). The expansion is naturally an anytime algorithm: it can be applied sequentially to CEGs in rank order, and stopped at any point when the time budget is exhausted, making BubbleRAG's precision-refinement phase controllable and cost-aware.

\subsection{Answer Generation}\label{sec:answer-gen}


Upon completion of reasoning-aware expansion, BubbleRAG merges the expanded subgraphs from the top-$n$ candidates into a single Unified Evidence Graph $G_{\text{final}}$, consolidating redundant nodes and edges that overlap across candidates. We then map nodes in $G_{\text{final}}$ back to their provenance in the source corpus (via the chunk pointers retained during indexing) and retrieve the corresponding original text chunks $C_{\text{text}}$. The assembled hybrid context, combining (i)~structured triples serialized from $G_{\text{final}}$ as a reasoning skeleton and (ii)~the associated raw text chunks as descriptive grounding, is fed into the LLM along with the original user query to produce the final answer.


\section{Experiments}

\begin{table*}[!htbp]
\centering
\caption{Comparison of Different Methods on Multi-hop QA Benchmarks. 
The best results are highlighted in bold. The rightmost two columns show the average F1 and ACC$_L$ across all datasets.}
\label{tab:main_results}
\renewcommand{\arraystretch}{1.25}
\setlength{\tabcolsep}{5pt}

\begin{adjustbox}{max width=\textwidth}
\begin{tabular}{c|rr|rr|rr|rr|rr|rr|rr}
\toprule
\multirow{2}{*}{\textbf{Method}} 
& \multicolumn{6}{c|}{\textbf{30B Model}} 
& \multicolumn{6}{c|}{\textbf{8B Model}} 
& \multicolumn{2}{c}{\textbf{Average}} \\ 
\cmidrule(lr){2-15}
& \multicolumn{2}{c|}{HotpotQA} & \multicolumn{2}{c|}{MuSiQue} & \multicolumn{2}{c|}{2Wiki} 
& \multicolumn{2}{c|}{HotpotQA} & \multicolumn{2}{c|}{MuSiQue} & \multicolumn{2}{c|}{2Wiki} 
& F1 & ACC$_L$ \\ 
\cmidrule(lr){2-15}
& F1 & ACC$_L$ & F1 & ACC$_L$ & F1 & ACC$_L$
& F1 & ACC$_L$ & F1 & ACC$_L$ & F1 & ACC$_L$ &  &  \\ 
\midrule
Vanilla LLM      & 17.84 & 29.70 & 5.80 & 10.40 & 20.14 & 30.30 & 11.41 & 30.10 & 5.18 & 12.90 & 17.46 & 29.90 & 12.97 & 23.88 \\
CoT + LLM        & 18.83 & 33.50 & 8.42 & 13.20 & 18.25 & 25.50 & 16.35 & 30.40 & 7.11 & 12.00 & 18.11 & 25.50 & 14.51 & 23.35 \\
Naive RAG        & 73.07 & 80.50 & 43.68 & 42.20 & 52.15 & 54.80 & 67.41 & 75.00 & 34.53 & 35.00 & 50.23 & 52.80 & 53.51 & 56.71 \\
ToG              & 40.07 & 48.50 & 16.76 & 21.60 & 35.88 & 38.60 & 23.56 & 35.30 & 8.39 & 13.90 & 20.22 & 23.10 & 24.15 & 30.16 \\
RAPTOR         & 66.75 & 68.00 & 33.10 & 32.20 & 40.11 & 43.90 & 62.66 & 68.60 & 29.35 & 28.90 & 36.25 & 40.20 & 44.70 & 46.97 \\
LightRAG (local)         & 57.05 & 61.80 & 41.72 & 43.80 & 60.17 & 64.90 & 54.25 & 54.30 & 36.57 & 37.70 & 54.61 & 55.11 & 50.73 & 52.94 \\
LightRAG (global)         & 35.28 & 39.10 & 20.56 & 22.50 & 23.20 & 26.30 & 30.79 & 30.80 & 20.11 & 20.30 & 23.41 & 24.00 & 25.56 & 27.17 \\
LightRAG (hybrid)         & 54.87 & 59.10 & 39.57 & 42.20 & 55.67 & 60.10 & 54.32 & 54.50 & 34.21 & 35.10 & 47.79 & 48.53 & 47.74 & 49.92 \\
\rowcolor{hippogreen}
HippoRAG2        & 71.72 & 80.20 & 45.04 & 47.20 & 67.65 & 71.70 & \textbf{73.33} & 79.00 & 42.92 & 41.60 & 62.40 & 66.70 & 60.50 & 64.40 \\
Clue-RAG         & 62.87 & 67.30 & 33.01 & 33.90 & 45.63 & 47.40 & 54.95 & 69.20 & 27.39 & 27.80 & 35.87 & 46.30 & 43.29 & 48.65 \\
\midrule
\rowcolor{bubbleblue}
\textbf{BubbleRAG} & \textbf{74.35} & \textbf{82.60} & \textbf{53.03} & \textbf{53.10} & \textbf{72.54} & \textbf{76.60} 
& 71.82 & \textbf{79.10} & \textbf{45.15} & \textbf{44.00} & \textbf{64.97} & \textbf{67.60} & \textbf{63.02} & \textbf{66.63} \\ 
\bottomrule
\end{tabular}
\end{adjustbox}
\end{table*}

\subsection{Experimental Setup}

\stitle{Baselines.}
To strictly evaluate the performance of BubbleRAG, we selected a diverse set of baselines categorized into four groups: general LLMs and standard retrieval methods, iterative multi-hop methods, random-walk methods, and pre-indexed structure methods.

\stitle{General LLMs and Standard Retrieval methods.} We first evaluate foundational models and standard pipelines. Vanilla LLM relies solely on parametric knowledge to answer queries without external retrieval. LLM + CoT\cite{2023-chainofthought} enhances this by utilizing a standardized Chain-of-Thought template to guide step-by-step reasoning before generating a final answer. NaiveRAG\cite{2024-surveyretrievalaugmentedgenerationl} represents the standard pipeline, retrieving text chunks based on vector embedding similarity and concatenating them as context for the LLM.

\stitle{Iterative multi-hop methods.} This paradigm constructs explicit reasoning chains across graph topologies. We evaluate ToG (Think-on-Graph) \cite{2024-thinkongraph}, which performs a constrained beam search over the knowledge graph to discover relational paths connecting query entities to potential answers.

\stitle{Random-walk methods} This category captures structural centrality using network algorithms. We evaluate HippoRAG2\cite{2025-ragmemory}, which integrates passage retrieval into the Personalized PageRank (PPR) algorithm \cite{1999-pagerank} to compute associative relationships between semantic query terms and document chunks.

\stitle{Pre-indexed structure methods.} We compare against frameworks that rely on pre-built hierarchical or multi-partite structures. RAPTOR \cite{2024-raptor} constructs a hierarchical tree to cluster and summarize documents, enabling retrieval across different levels of abstraction. ClueRAG builds a static multi-partite graph index to model interactions between questions, documents, and concepts, utilizing the LLM to navigate this structure. LightRAG \cite{2025-lightrag} indexes documents into a comprehensive graph of entities and relations. We assess its three distinct retrieval modes: LightRAG-L (Local) for low-level entity details, LightRAG-G (Global) for broader conceptual queries, and LightRAG-H (Hybrid), which combines both strategies.

\stitle{Datasets.}
We evaluated BubbleRAG on three benchmark datasets designed for multi-hop Question Answering (QA): MuSiQue\cite{2022-musique}, HotpotQA\cite{2018-hotpotqa}, and 2WikiMultiHopQA \cite{2020-multihopqa}. These datasets require aggregating information from multiple documents to derive the correct answer. Following established protocols \cite{2024-hipporag,2022-musique}, we used the validation split of each dataset (consisting of 1,000 questions per dataset) for evaluation. The corresponding corpora were preprocessed into chunks to serve as the retrieval source.

\stitle{Metrics.}
Consistent with previous research\cite{2024-hipporag,2025-cluerag}, we employ two primary metrics to assess performance. First, we use the F1 Score to measure the lexical overlap between the generated answer and the ground truth, assessing the precision and recall of the exact tokens appearing in the prediction. Second, since multi-hop questions often have valid answers that differ in phrasing (e.g., synonyms or sentence structures), exact matching can be misleading. To address this, we employ LLM-as-a-Judge Accuracy ($ACC_L$). We utilize Qwen-3-7B\cite{2025-qwen3} as a judge to evaluate the semantic equivalence between the model's output and the gold standard, assigning a score of 1 for correct answers and 0 otherwise.

\stitle{Implementation Details.}
To ensure a fair comparison across all methods, we use Qwen3-Embedding-8B\cite{2025-qwen3embedding} embedding model for every framework in our experiments. In addition, following common fairness settings in prior work~\cite{2025-cluerag}, we impose a unified retrieval budget: each method is allowed to observe at most $15$ text chunks per query. This constraint prevents performance gains that stem purely from larger contexts and ensures that improvements reflect retrieval quality rather than input size. Unless otherwise specified, all other hyperparameters follow the default configurations recommended in the original implementations.

\subsection{Main Results}

Table \ref{tab:main_results} presents the comparative performance of BubbleRAG and ten baseline methods across three multi-hop QA benchmarks using both 30B and 8B parameter models. The results demonstrate the consistent superiority of BubbleRAG, which achieves the highest average F1 and LLM-as-a-Judge Accuracy scores across all settings. Specifically, on the 30B model setting, BubbleRAG surpasses the strongest baseline, HippoRAG2, by an average of 2.52\% in F1 and 2.23\% in accuracy. BubbleRAG's superior F1 score reflects its balanced optimization of recall (via group-covering expansion) and precision (via discriminative ranking).

The advantage of BubbleRAG is particularly pronounced on the MuSiQue dataset, which requires complex multi-hop reasoning. While other structure-aware baselines like Clue-RAG and LightRAG struggle to surpass Naive RAG on this dataset, BubbleRAG achieves an F1 score of 53.03 with the 30B model, outperforming HippoRAG2 by approximately 8 percentage points. MuSiQue requires 3--4 hop reasoning, where existing methods' single-anchor or fixed-hop strategies break down. BubbleRAG's group-aware expansion naturally handles variable-length reasoning chains without predefined hop limits. Unlike approaches prone to error propagation, the heuristic Steiner Tree search in BubbleRAG ensures semantically grounded and complete reasoning chains.

Furthermore, BubbleRAG maintains its leading position across model scales. Notably, the method achieves an average F1 of 63.02 using the smaller 8B model, which is comparable to, and often better than, the performance of many baselines using larger 30B models. This demonstrates that retrieval quality, not model size, is the primary bottleneck for multi-hop QA: the high-quality context provided by BubbleRAG effectively compensates for the reduced parametric knowledge of smaller models.

\subsection{Ablation Studies}

To validate the contribution of each component in BubbleRAG and analyze the impact of key hyper parameters, we conducted comprehensive ablation studies and sensitivity analyses. All ablation experiments use Qwen3-8B\cite{2025-qwen3} as the backbone LLM to ensure fair comparison.

\stitle{w/o Anchor Specialization:} To evaluate the importance of our query-conditioned constraint refinement, we replaced the LLM-based rewriting module with a direct keyword extraction approach. In this variant, we extract explicit entities from the query using a standard extractor and map them to graph nodes without refining generic terms (e.g., \textit{mother}) into specific constraints (e.g., \textit{Lothair II's mother}). This setting tests the system's ability to handle ambiguity without query-specific context injection.

\stitle{w/o Schema Relaxation:} To assess the contribution of the chunk-guided schema relaxation mechanism, we retained the constraint refinement but removed the retrieval of text chunks prior to anchor selection. In this setting, the system relies strictly on the refined keywords to locate anchors in the knowledge graph, without the flexibility to relax schema-sensitive conditions based on local text evidence.

\stitle{w/o CEG Ranking:} To verify the necessity of the discriminative ranking mechanism, we removed the scoring function defined in Eq.~(\ref{eq:score}). Instead of re-evaluating the candidate subgraphs based on semantic dissonance and structural completeness, we simply selected the first valid connected subgraph returned by the heuristic expansion. This variant highlights the role of the ranking module in filtering out structurally connected but semantically irrelevant noise.

\begin{table}[t]
    \centering
    \caption{Ablation study of BubbleRAG components on the 2Wiki and HotpotQA datasets. We report the F1 score and LLM-as-a-Judge Accuracy (ACC). The \textit{Full Method} row represents the standard BubbleRAG performance.}
    \label{tab:ablation_results}
    \begin{tabular}{lcccc}
    \toprule
    \multirow{2}{*}{\textbf{Method Variant}} & \multicolumn{2}{c}{\textbf{2Wiki}} & \multicolumn{2}{c}{\textbf{HotpotQA}} \\
    \cmidrule(lr){2-3} \cmidrule(lr){4-5}
     & \textbf{F1} & \textbf{ACC} & \textbf{F1} & \textbf{ACC} \\
    \midrule
    \textbf{BubbleRAG (Full Method)} & \textbf{64.97} & \textbf{67.60} & \textbf{71.82} & \textbf{79.10} \\
    \midrule
    w/o Anchor Specialization & 60.45 & 65.00 & 69.28 & 76.10 \\
    w/o Schema Relaxation & 53.62 & 59.20 & 66.65 & 75.30 \\
    w/o CEG Ranking & 58.76& 64.40 & 70.20 & 78.90 \\
    \bottomrule
    \end{tabular}
\end{table}

\stitle{Ablation Analysis:}
The ablation experiments are conducted on the full validation sets of 2Wiki and HotpotQA (1,000 questions each). As shown in Table~\ref{tab:ablation_results}, removing each component leads to measurable performance degradation, validating the importance of our design choices. The most critical component is:

\sstitle{Schema Relaxation}, whose removal causes the largest performance drop: F1 decreases by 11.35 points on 2Wiki and 5.17 points on HotpotQA. This confirms that chunk-guided anchor selection effectively bridges the gap between schema-level graph entities and textual evidence, preventing the system from anchoring on semantically mismatched nodes.

\sstitle{Anchor Specialization} also demonstrates significant value, with F1 drops of 4.52 and 2.54 points on 2Wiki and HotpotQA respectively. This indicates that query-conditioned constraint rewriting is essential for resolving entity ambiguity. Without refining generic terms, the system fails to locate the correct anchor nodes in the knowledge graph.

\sstitle{CEG Ranking} contributes to performance, though its impact varies across datasets. On 2Wiki, removing it causes a 6.21-point F1 drop, while on HotpotQA the degradation is relatively minor (1.62 points). This suggests that the ranking module is particularly crucial for complex reasoning tasks requiring precise evidence selection, while simpler datasets can tolerate some noise in the retrieved context.

In summary, removing Anchor Specialization and Schema Relaxation primarily reduces \emph{recall} (the system fails to find correct anchors), while removing CEG Ranking reduces \emph{precision} (noisy candidates are selected). This validates BubbleRAG's design of addressing recall and precision through dedicated pipeline stages.

\subsection{Parameter Sensitivity}

We further analyze how the system's performance and efficiency fluctuate with key hyperparameters: the expansion budget $B$, the reasoning depth $d$, and the penalty factor $\alpha$. The localization hop $h$ is held constant at $h=6$ throughout, as preliminary experiments showed this value consistently provides sufficient local coverage for the datasets studied.

\begin{table}[h]
    \centering
    \caption{Sensitivity analysis on the 2Wiki dataset. We report F1 scores and average Inference Latency (Lat.) in seconds. The default settings are $B=10, h=6, d=6, \alpha=1.0$.}
    \label{tab:sensitivity}
    \renewcommand{\arraystretch}{1.1}
    \setlength{\tabcolsep}{8pt}
    \begin{tabular}{lccc}
    \toprule
    \textbf{Parameter} & \textbf{Value} & \textbf{F1 Score} & \textbf{Lat. (s)} \\
    \midrule
    \multirow{4}{*}{Budget ($B$)}
      & 5  & 59.48 & 18.72 \\
      & 10 & 60.52 & 20.99 \\
      & 20 & 58.32 & 33.67 \\
      & 50 & 60.42 & 44.98 \\
    \midrule
    \multirow{4}{*}{Depth ($d$)}
      & 2 & 58.25 & 16.17 \\
      & 4 & 59.07 & 19.39 \\
      & 6 & 60.52 & 20.99 \\
      & 8 & 62.28 & 32.67 \\
    \midrule
    \multirow{4}{*}{Penalty ($\alpha$)}
      & 0.1 & 59.11 & -- \\
      & 1.0 & 60.52 & -- \\
      & 2.0 & 60.37 & -- \\
      & 5.0 & 57.35 & -- \\
    \bottomrule
    \end{tabular}
\end{table}

\stitle{Sensitivity Analysis:}
The parameter sensitivity analysis is conducted on a randomly sampled subset of 100 queries from 2Wiki. As shown in Table~\ref{tab:sensitivity}, this analysis reveals the impact of key hyperparameters on both performance and efficiency.

\sstitle{Expansion Budget ($B$):} The budget parameter controls the number of candidates explored during bubble expansion. We observe that $B=10$ achieves an optimal balance, yielding an F1 of 60.52 with a moderate latency of 20.99 seconds. Increasing $B$ to 20 or 50 does not yield meaningful performance gains (F1 drops to 58.32 and 60.42 respectively) but significantly increases latency to 33.67s and 44.98s. This indicates that the bubble expansion algorithm efficiently identifies high-quality candidates early, and excessive exploration introduces noise without improving answer quality.

\sstitle{Reasoning Depth ($d$):} As the reasoning depth increases from 2 to 8, we observe consistent F1 improvements: 58.25 $\rightarrow$ 59.07 $\rightarrow$ 60.52 $\rightarrow$ 62.28. However, the latency grows substantially, particularly when $d=8$ (32.67s). This suggests that deeper reasoning chains are beneficial for complex multi-hop questions but come with a significant computational cost. For practical deployment, $d=6$ offers the best trade-off between accuracy and speed.

\sstitle{Penalty Factor ($\alpha$):} The penalty factor in Eq~(\ref{eq:penalty}) controls the strictness of the structural completeness requirement: higher $\alpha$ penalizes CEGs that miss anchor groups more severely. We observe that the performance peaks at $\alpha = 1.0$ (yielding an F1 of 60.52) but experiences a noticeable drop when $\alpha$ is increased to 5.0 (F1 drops to 57.35). This indicates that while a moderate completeness enforcement is necessary to capture diverse query concepts, an overly strict requirement ($\alpha=5.0$) hurts performance. Enforcing a rigid AND-style conjunction can overly penalize highly relevant candidate graphs that simply miss an unimportant anchor.

\subsection{Efficiency \& Cost Analysis}
\label{sec:efficiency}

Given that BubbleRAG involves a multi-stage pipeline with LLM interactions (Steps 2 \& 5) and heuristic graph algorithms (Steps 3--4), we evaluate its computational overhead against representative baselines. We conducted this analysis on a randomly sampled subset of 100 queries to measure the average build time, inference latency, and token consumption.

\begin{table}[h]
    \centering
    \caption{Efficiency comparison on 100 sampled queries (single A100 GPU). Latency: average seconds per query. Query: total tokens consumed across all 100 queries at inference time. Index: total tokens for offline corpus indexing}
    \label{tab:efficiency}
    \renewcommand{\arraystretch}{1.1}
    \setlength{\tabcolsep}{8pt}
    \begin{tabular}{lccc}
    \toprule
    \textbf{Method} & \textbf{Latency (s)} & \textbf{Query} & \textbf{Index} \\
    \midrule
    Naive RAG & \textbf{0.67} & 249,476 & - \\
    HippoRAG2 & 4.26 & 418,812 & 4,575,580\\
    ToG & 45.93 & 765,915 & - \\
    \textbf{BubbleRAG} & 20.99 & 1,064,052 & 3,840,320 \\
    \bottomrule
    \end{tabular}
\end{table}

Table~\ref{tab:efficiency} presents a comprehensive efficiency comparison. While Naive RAG has the lowest latency (0.67s) due to simple vector retrieval, its performance is significantly inferior. BubbleRAG (20.99s) is substantially faster than ToG (45.93s) while delivering superior accuracy, demonstrating the efficiency of our heuristic bubble expansion algorithm compared to ToG's iterative LLM-based graph traversal.

Compared to HippoRAG2 (4.26s), BubbleRAG incurs higher latency. The $\approx$5$\times$ latency difference is largely attributable to the LLM interactions in anchor grouping (Step~2) and reasoning-aware expansion (Step~5). This overhead is justified by its consistent performance gains across all benchmarks. The token consumption analysis reveals that BubbleRAG (1,064,052 tokens) uses more tokens than baselines, which is expected as it encodes graph structure into prompts. However, this cost is offset by the higher answer quality, and the CEG Ranking phase ensures only relevant evidence reaches the generation stage, avoiding wasteful processing of noisy contexts.

Furthermore, BubbleRAG's index construction cost (3,840,320 tokens) is comparable to HippoRAG2 (4,575,580 tokens) and is amortized across all queries, making it practical for real-world deployments where the knowledge base is queried repeatedly.

\section{Related Work}

We review existing approaches for knowledge graph retrieval.

\stitle{Query Rewriting and Schema-Aligned Matching.}
Methods in this paradigm, such as SimGRAG \cite{2025-simgrag} and KG-GPT \cite{2023-kggpt}, align queries with pre-conceived structural templates. SimGRAG utilizes a two-stage process where LLM transforms a query into an abstract pattern graph for subgraph alignment. Conversely, KG-GPT employs a divide-and-conquer strategy, segmenting complex queries into discrete \textit{(head, relation, tail)} triple structures to retrieve candidate subgraphs. These methods rely on the LLM's internal prior knowledge. When generated patterns do not exist in the underlying topology, it results in hallucinations and retrieval gaps.

\stitle{Iterative Multi-Hop Exploration.}
This paradigm initiates retrieval at semantic seed nodes and expands outward to construct reasoning chains. ToG \cite{2024-thinkongraph} implements a constrained beam search directly on the topology, while ToG2 \cite{2025-thinkongraph2} incorporates document embeddings to filter contextually irrelevant nodes during traversal. RoG \cite{2024-reasoninggraphs} generates a discrete sequence of relation types as a search plan, executing a constrained breadth-first search to extract grounded paths. These mechanisms are highly sensitive to initial anchor extraction; a single misalignment at the seed node stage causes cascading retrieval failures across the multi-hop chain.

\stitle{Stochastic Traversal via Random Walks.} Methods in this category, such as HippoRAG, HippoRAG2 \cite{2024-hipporag, 2025-ragmemory}, LinearRAG \cite{2025-linearrag}, and AGRAG \cite{2025-agrag}, utilize algorithms like Personalized PageRank (PPR) \cite{1999-pagerank} from semantic anchors to capture structural centrality and associative memory. HippoRAG extracts named entities from the query as seed nodes, executing PPR over an open knowledge graph to distribute probability mass and rank passages. HippoRAG2 extends this by integrating passage nodes directly into the graph topology alongside phrase nodes, assigning PPR reset probabilities based on embedding similarity. LinearRAG constructs a relation-free hierarchical graph of entity, sentence, and passage nodes, applying PPR for global importance aggregation from locally activated seed entities. AGRAG computes PPR to establish node influence scores and combines them with semantic edge costs to extract a Minimum Cost Maximum Influence (MCMI) reasoning subgraph. Because probabilistic walks necessitate a starting point, their retrieval performance is upper-bounded by the accuracy of the initial anchor extraction; erroneous anchors will misguide the subsequent topological exploration.

\stitle{Structure-Augmented Retrieval with Auxiliary Graphs.} Methods like GraphRAG \cite{2024-edge2024local}, KAG \cite{2024-kag}, and ClueRAG \cite{2025-cluerag} rely on pre-processing to construct hierarchical indices or community structures. GraphRAG extracts an entity knowledge graph from source documents, applies the Leiden algorithm to partition the graph into a nested hierarchy of modular communities, and uses a language model to generate bottom-up summaries for these partitions. KAG constructs a mutual-index between the graph structure and raw text chunks using a hierarchical framework, applying semantic reasoning during offline indexing to complete concept relations and align fragmented instances. ClueRAG builds a multi-partite graph index comprising three distinct node layers: text chunks, knowledge units, and entities. Because these static, pre-built structures are domain-dependent and constructed independently of user queries, they are poorly suited for the diverse and specific requirements of dynamic real-world queries.

BubbleRAG differs fundamentally from all four paradigms. Unlike query rewriting methods, it does not assume prior schema knowledge or rely on the LLM's ability to generate structurally valid graph patterns. Unlike iterative multi-hop and random-walk methods, it initializes from \emph{groups} of anchors rather than a single seed node, avoiding the cascading failure mode inherent in single-anchor strategies. Unlike pre-indexed structure methods, its evidence structures are query-specific and dynamically constructed at retrieval time. Most distinctively, BubbleRAG \emph{jointly optimizes} for recall (through group-covering bubble expansion) and precision (through discriminative CEG ranking), treating retrieval as a formal optimization problem (OISR) rather than a traversal heuristic.

\section{Conclusion}

We addressed the problem of retrieval from black-box knowledge graphs, where the retriever has no access to the graph's schema. We formalized this as the Optimal Informative Subgraph Retrieval (OISR) problem, a variant of Group Steiner Tree that captures the dual goals of semantic coverage (recall) and information density (precision), and proved it to be NP-hard and APX-hard. Based on this formalization, we proposed BubbleRAG, a training-free pipeline that systematically addresses recall through semantic anchor grouping and heuristic bubble expansion, and precision through candidate evidence graph ranking and reasoning-aware expansion.

Comprehensive experiments on multi-hop QA benchmarks demonstrate that BubbleRAG consistently outperforms state-of-the-art baselines, including HippoRAG2, achieving superior accuracy and F1 scores. Notably, our framework exhibits exceptional robustness, delivering competitive performance even with smaller language models (8B) compared to baselines using larger models (30B). BubbleRAG's plug-and-play design requires no retriever fine-tuning or modifications to the underlying KG structure, and its localized subgraph construction ensures scalability to massive knowledge graphs.

\balance
\bibliographystyle{ACM-Reference-Format}
\bibliography{reference}

\end{document}